\documentclass[ USenglish, a4paper]{lipics-v2021}

\pdfoutput=1 
\hideLIPIcs  
\nolinenumbers

\usepackage{fontenc}
\usepackage{todonotes}

\usepackage{booktabs}
\usepackage{tcolorbox}
\usepackage{multirow}
\usepackage{graphicx}
\usepackage{nicefrac}
\usepackage{pdflscape}
\usepackage{xcolor}

\usepackage{hyperref}
\usepackage{graphicx}
\usepackage{enumerate}
\usepackage{tikz}
\usetikzlibrary{backgrounds,calc}
\usepackage{float}
\usepackage{color}
\usepackage{acronym}
\usepackage{multirow}

\usepackage{wrapfig}
\usepackage{libertine}

\usepackage{algorithm}
\usepackage{algpseudocode}

\usepackage{amsmath}
\usepackage{flexisym}
\usepackage{breqn}

\usepackage{mdframed}
\usepackage{apxproof}
\newtheoremrep{apxlemma}{Lemma}

\newcommand{\pr}[1]{\ensuremath{\text{{\bf Pr}$\left[#1\right]$}}}
\newcommand{\E}[1]{\ensuremath{\text{{\bf E}$\left[#1\right]$}}}

\title{Beep-And-Sleep: Message and Energy Efficient Set Cover}

\author{Thorsten Götte}{ Paderborn University, Germany}{thgoette@mail.upb.de}{}{}
\author{Christina Kolb}{ University of Twente, The Netherlands}{c.kolb@utwente.nl}{}{}
\author{Christian Scheideler}{Paderborn University, Germany}{scheidel@mail.upb.de}{}{}
\author{Julian Werthmann}{ Paderborn University, Germany}{jwerth@mail.upb.de}{}{}
\authorrunning{Götte et al.} 



\begin{document}
\Copyright{Thorsten Götte, Christina Kolb, Christian Scheideler, and Julian Werthmann}
\maketitle

\begin{abstract}

We observe message-efficient distributed algorithms for the \textsc{SetCover} Problem.
Given a ground set $U$ of $n$ elements and $m$ subsets of $U$, we aim to find the minimal number of these subsets that contain all elements.
In the default distributed distributed setup of this problem, each set has a bidirected communication link with each element it contains.
Our first result is a $\tilde{O}(\log^2(\Delta))$-time and $O(\sqrt{\Delta)}(n+m))$-message algorithm with expected approximation ration of $O(\log(\Delta))$ in the $KT_0$ model. 
The value $\Delta$ denotes maximal cardinally of each subset.
Our algorithm is \emph{almost} optimal with regard to time and message complexity.
Further, we present \textsc{SetCover} algorithm in the \textsc{Beeping} model that only relies on carrier-sensing and can trade runtime for approximation ratio similar to the celebrated algorithm by Kuhn and Wattenhofer [PODC '03].
\end{abstract}

\section{Introduction}
\textsc{SetCover} is a well-understood problem in both the centralized and distributed setting. 
Given a collection of elements $\mathcal{U} := \{e_1, \ldots, e_n\}$ and sets $\mathcal{S} := \{s_1, \ldots, s_m\}$ with $s_i \subseteq \mathcal{U}$ the goal is to cover all elements with as few sets as possible.
\textsc{SetCover} has a wide variety of applications in many areas of computer science.
On the one hand, it plays an essential role in the analysis of large data sets, which is often needed in fields like operations research, machine learning, information retrieval, and data mining (see \cite{DBLP:conf/soda/IndykMRVY18} and the references therein). 
On the other hand, it is also used in purely distributed domains like ad-hoc sensor networks.
An essential task in these networks is to determine a minimum set of nodes, a so-called \textsc{DominatingSet}, such that all nodes are sensor range of this set. 
This set can then fulfill particular tasks like routing, collecting sensor data from neighbors, and various other tasks.
Note that \textsc{DominatingSet} is a special case of \textsc{SetCover} where all sets are also elements.

In the centralized setting, a simple greedy algorithm that always picks the set that covers most elements has a logarithmic approximation factor, which is the best we can hope for a polynomial-time algorithm unless $P=NP$ .
There are several randomized algorithms in the distributed \textsc{Congest} model that match the optimal approximation ratio, w.h.p., and have a near-optimal runtime of $O(\log(\Delta)^2)$ where $\Delta$ is the maximum degree of a set \cite{DBLP:journals/dc/JiaRS02,DBLP:conf/podc/KuhnW03}.
For distributed algorithms, an instance of \textsc{SetCover} is ususally modeled as a so-called \emph{problem graph} $G_P := \{V_{\mathcal{U}} \cup V_{\mathcal{S}}, E\}$. 
Each set and each element corresponds to a node in this graph and for each set $s_i \in \mathcal{S}$ there is a bidirected edge $\{e_j,S_i\} \in E$ to each element $e_j \in \mathcal{U}$ it contains.
Each of these edges models a bidirected communication channel between the nodes representing the set and element in the default distributed setup.
Each set and element in each round of communication can send a distinct message of size $O(\log(n))$ over each communication edge.
Although there are numerous results concerning the runtime of algorithms in the CONGEST model, there are only a few that consider the message complexity, i.e., how many messages are needed to compute a distributed solution to \textsc{SetCover}.
    To the best of our knowledge, in all popular solutions to set cover, all sets/elements extensively communicate with their neighborhood and send (possibly distinct) messages to all of their neighboring set and elements.
    This setup implies a message complexity of $O(|E|)$ per round and $\tilde{O}(|E|)$ overall.
    From a practical standpoint, a high message complexity makes these algorithms less suitable for dense networks where sending a single message is considered costly. 
    For example, in ad-hoc networks, one needs to take special care of the limited battery life of the nodes.
    In particular, sending messages should be reduced to a minimum due to the high energy requirements of the radio module.
    Thus, a distributed message-efficient algorithm for \textsc{SetCover} can therefore greatly reduce this energy consumption.
    Therefore, our first question is:
    \begin{align*}
        (Q.1) \textit{  Can we find a \textsc{SetCover} algorithm that sends $o(|E|)$ messages?}
    \end{align*}
    However, even if we find message-efficient algorithms in the CONGEST model, they do not necessarily translate well to ad-hoc networks where they are needed. 
    In practical ad-hoc networks, however, it is much more energy-efficient to only send a single signal, a so-called \textsc{Beep}, to all neighboring nodes.
    Further, nodes can only distinguish if at least one or none of their neighbors beeped, i.e., they only rely on carrier sensing. 
    This has recently been formalized in the so-called \textsc{Beeping} model\cite{DBLP:conf/wdag/CornejoK10}.
    To the best of our knowledge there no \textsc{SetCover} or \textsc{DominatingSet} algorithm \emph{for general graphs} in this model.
    Thus, our second question is:
    \begin{align*}
        (Q.2) \textit{  Can we find a \textsc{SetCover} algorithm in the BEEPING model?}
    \end{align*}
We answer both questions affirmative and provide fast and efficient algorithms for \textsc{SetCover} in the $KT_0$-model (a variant of the CONGEST model for analyzing message-efficient algorithms) and the \textsc{Beeping} model. Our $KT_0$ algorithm has polylogarithmic runtime and approximation ratio while sending only $\tilde{O}(\sqrt{\Delta}(n+m))$\footnote{$\tilde{O}(\cdot)$ hides polylogarithmic factors.} messages, w.h.p\footnote{As usual, an event holds with high probability, if it holds with prob. $1-o(n^c)$ for some $c>1$}. Although we cannot prove that this message complexity is optimal for this particular approximation ratio, we give evidence that one cannot hope for far better results as we show that there are instances that require $O(\sqrt{\Delta}n)$ messages for a constant approximation. 
Second, our \textsc{Beeping} algorithm runs in $O(k^3)$ rounds and has an approximation ratio of $O(\log(\Delta)^2\sqrt[k/3]{\Delta})$ and thus --- similar to \cite{DBLP:conf/podc/KuhnW03} --- can trade runtime for approximation ratio. This tradeoff makes it particularly practical as the topology of sensor networks is often rapidly changing, and our algorithm can quickly react to these changes by computing a new solution in constant time.
Existing solutions (\cite{DBLP:conf/mobihoc/ScheidelerRS08,DBLP:conf/infocom/YuJYLC15}) have a far better approximation ratio (as they are designed for UDGs), but have at least logarithmic runtime, which cannot trivially be reduced.
Note that the focus of this work is to prove feasibility and not optimize logarithmic factors in message complexity and approximation ratios, i.e., we explicitly do not claim our factors to be optimal.

Our main technical contribution is an adapted version of Jia et al.'s. \textsc{DominatingSet} algorithm from \cite{DBLP:journals/dc/JiaRS02} that will be the basis for both our contributions. We adapt this algorithm in two major ways: First, we massively reduce the messages that need to be exchanged by replacing all instances where nodes are counted through randomized approximation that only use a fraction of the messages \emph{and} --- perhaps more importantly --- completely replacing the mechanism that lets nodes decide to join the solution. For the latter, we use geometrically distributed starting times where nodes add themselves when their time has come, and they have a certain threshold of uncovered neighbors. This approach does not require any additional messages. Note that the independent work of \cite{DBLP:conf/soda/GrunauMRV20} which only recently came to our attention, uses a very similar technique\footnote{Instead of using geometric starting times they continuously increase a set's probability to join until it joins or does cover enough elements. From a probabilistic point of view this is (almost) equivalent to picking geometric starting times.}, but their analysis does not provide all the properties we need for our problems, at least not without non-trivial extensions.

This paper is structured as follows: First, we present some related work in the remainder of this section. Then, we present the algorithm for the \textsc{Beeping} model in Section \ref{sec:BEEP} and in Section \ref{sec:KT1} the message-efficient $KT_0$ algorithm and an almost matching lower bound. We provide more details on the models in the respective sections.

\subsection{Related Work}
There are only a few works on the message complexity of \textsc{SetCover}, however
the space complexity of \textsc{SetCover} in streaming models is an active research topic with exciting results for several variants of the problem.
In the streaming model, the edges of the problem graph, i.e., the information on which element belongs to which set, arrive sequentially.
An algorithm can store each edge for later use and iterate over the whole input several times.
The goal is to solve \textsc{SetCover} using as few passes and as few space as possible.
Note that a trivial algorithm can just iterate over the input once, store all edges, and then solve the problem using an \emph{offline} algorithm.
Demaine et al. further showed that randomization is crucially needed for a space complexity of $o(mn)$ as it is impossible for a deterministic algorithm.
In addition to this distributed model, Indyk et al. also successfully used the ideas and primitives from streaming algorithms in the area of Sublinear Algorithms \cite{DBLP:conf/soda/IndykMRVY18}.
The goal is to find an algorithm that makes as \emph{few} queries as possible and computes a good approximate solution.
As we will see, this model has a lot in common with our distributed model.
In this case, the number of queries should be $o(mn)$ , i.e., sublinear in the number of edges. Indyk et al. present a polynomial-time algorithm that makes $\tilde{O}(m+\sqrt{m}n)$ queries and computes a $O(\log^2(n))$-approximate solution. 
With exponential runtime, the approximation ratio even reduces to $O(\log^2(n))$, but the algorithm still needs to make $\tilde{O}(m+\sqrt{m}n)$ queries.
Furthermore, Indyk et al. gave strong evidence that $\Omega(nm^\epsilon)$ queries are indeed necessary to achieve good approximation ratios. 
Finally, we review some fully distributed approaches for \textsc{SetCover} and related Optimization Problems.
As mentioned in the introduction Jia et al.\cite{DBLP:journals/dc/JiaRS02} as well as Kuhn and Wattenhofer\cite{DBLP:conf/podc/KuhnW03} have presented fast distributed algorithms to solve the \textsc{DominatingSet} problem in time in $O(\log(n)^2)$. The goal is to find the minimal set of nodes adjacent to all nodes in the graph $G:=(V,E)$. 
\textsc{DominatingSet} is a special case of \textsc{SetCover} where each node is both an element and a set.
Kuhn et al. further presented an algorithm to solve SetCover (and any covering and packing LP) in $O(\log(n)^2)$ rounds in the CONGEST model\cite{DBLP:conf/soda/KuhnMW06}.
This is close to the optimal runtime of $O(\log(n))$, which is also proven in that paper.
Roughgarden et al. presented an distributed algorithm for convex optimization problems (which includes \textsc{SetCover}) that works on any eventually connected communication network \cite{DBLP:journals/siamjo/Mosk-AoyamaRS10}.
In particular, it works on dynamic communication networks that can change from round to round.
Even though this algorithm is fully distributed, it heavily relies on sequential calculations to cope with these general communication networks.
Therefore, its runtime and message complexity are polynomial in the number of nodes.
Finally, many works consider the \textsc{DominatingSet}-problem in models tailored to ad-hoc networks.
First, there is a paper by Scheideler et al.\cite{DBLP:conf/mobihoc/ScheidelerRS08} that observes the SINR model where the nodes are modeled as points in the euclidean plane.
Finally, there already is a solution to the \textsc{DominatingSet}-problem in the \textsc{Beeping}-model for unit-disk graphs\cite{DBLP:conf/infocom/YuJYLC15}.
Their algorithm bears some similarities with ours but is not applicable to general graphs. 

\section{Preliminaries}
In addition to some textbook probability distributions, in this work, we use the following version of the Chernoff Bound:
\begin{lemma}[Chernoff Bound]
Let $X := \sum X_i$ be the sum of independent random variables with $X_i \in \{0,1\}$.Then it holds that for any $0<\delta$,
\begin{align}
    	\pr{X \geq (1-\delta)\E{X}} \leq e^{-\frac{\delta^2\E{X}}{3}}
 \,\,\, \textit{and} \,\,\, \pr{X \leq (1-\delta)\E{X}} \leq e^{-\frac{\min\{\delta,\delta^2\}\E{X}}{3}}.
\end{align}
\end{lemma}

\section{An Efficient \textsc{SetCover}-Algorithm for the \textsc{Beeping}-Model}
\label{sec:BEEP}
We will now describe our first algorithm, which we dub the $\textsc{Beep-And-Sleep}$ algorithm, as most sets and elements will be idle during the execution.
Before we go into the details of our result and the algorithm, let us first present the model.
In this section, we use the following (standard) variant of the \textsc{Beeping} model \cite{DBLP:conf/wdag/CornejoK10,dufoulon_et_al:LIPIcs:2018:9809}:
\begin{enumerate}
    \item We observe a fixed communication graph $G := (V_S \cup V_{\mathcal{U}} ,E)$ with $V_S=n$ and $V_{\mathcal{U}} = m$. Each set $s \in V_S$ has a bidirected edge $\{s,e\} \in E$ to each element $e \in V_{\mathcal{U}}$ it contains. Each node can only communicate with its neighbors in $G$. Further, all nodes know $\Delta$, the maximal degree of $G$. This assumption can be replaced by a polynomial upper bound, which would slow the algorithm down by a constant factor. Note that the nodes do not know their exact degree, and nodes have \emph{no} identifiers.
    \item Time proceeds in so-called \emph{slots}. In each slot, a node can either \emph{beep} or \emph{listen}. If a node \emph{listens} and any subsets of its neighbors \emph{beeps}, the \emph{listening} node receives a \textsc{Beep}. It can neither distinguish which neighbors beeped nor how many neighbors beeped, i.e., it only relies on carrier sensing.  
    Further, a node \emph{cannot} simultaneously beep and listen but must choose one of the two options.
    \item All nodes wake up in the same slot, i.e., we observe the \textsc{Beeping} model with simultaneous wake-up. We believe that our algorithm can also be extended for arbitrary wake-up as each node only needs to be in sync with neighboring nodes. If nodes do not wake up in the same round, their internal counters only differ by $1$ as each node wakes up one slot after its earliest neighbor. In this case, there are some standard tricks to simulate a single slot of a simultaneous wake-up algorithm within $3$ slots \cite{DBLP:conf/wdag/CornejoK10,dufoulon_et_al:LIPIcs:2018:9809}.
\end{enumerate}
Given this model, we show the following:
\begin{theorem}
There is an algorithm in the \textsc{Beeping} model that solves \textsc{SetCover} in time $O(k^3)$ with approximation ratio $O(\log^2(\Delta) \cdot \sqrt[k]{\Delta}^3)$ where $k > 3$ is parameter known to all nodes.
\end{theorem}
Thus, even for a constant $k$, our algorithm achieves a non-trivial approximation ratio, which is close to the optimal one as Kuhn et al. showed than any distributed algorithm with only local communication needs $O(k)$ for an approximation ratio of $O(\sqrt[k]{\Delta})$ \cite{DBLP:conf/podc/KuhnW03}.
We will now describe the algorithm promised by the theorem:
The core idea is that all nodes that cover the \emph{most} neighbors add themselves to the set similar to the sequential greedy algorithm.
The main difficulty stems from the fact that we need to estimate the number of uncovered neighbors correctly and to avoid that too many sets that cover the same elements concurrently add themselves to the solution.
The algorithm runs in $k$ phases, where in phase $i$ all sets that cover \emph{approximately} $\Delta_i = \frac{\Delta}{\sqrt[k]{\Delta}^i}$ elements try to add themselves to the solution.  
Each phase is structured in $4k$ \emph{rounds}. Further, each round again consists of $4k+1$ slots, which brings the total runtime to $O(k^3)$. 
The pseudocode for the algorithm is given in Figure $1$, on a high level, it works as follows:
\begin{enumerate}
    \item At the beginning of \emph{phase} $j$, i.e., before the first round of that phase, the sets and elements do the following:
     Each \emph{set} $s$ draws a geometric random variable $X_s$ with parameter $1-\nicefrac{1}{\sqrt[k]{\Delta}}$. 
    Values bigger than $4k$ are rounded down to $4k$, so $X \in [0,4k]$ always. 
    Then, the set waits for $4k-X_s$ \emph{rounds} and neither beeps or listens, it just stays idle.
     Each uncovered \emph{element} marks itself as \emph{active} with probability $\frac{4k}{\Delta_j} := \frac{4k\sqrt[k]{\Delta}^j}{\Delta}$. 
     
    Thus, each set with $\Delta_j$ (or more) uncovered elements has at least $4k$ active elements in expectation. 
    Each active element further picks a slot number $a_u \in [0,4k]$ uniformly at random.
\item In each \emph{round} $i$ of phase $j$, the sets and elements do the following:
     Each \emph{set} $s$ that wakes up in this round listens to \textsc{Beeps} for the first $4k$ slots of this round.
    Further, it counts all slots in which it received a \textsc{Beep}. 
    Otherwise, it remains idle for this round.
    Each active \emph{element} that has not yet been covered, beeps \emph{only} in the slot $a_u$ is has drawn in the beginning and remains idle otherwise.
    Note that, by our choice of active elements, this implies that all sets that have at least $\Delta_j$ uncovered elements, therefore, hear one \textsc{Beep} per slot 
    \emph{in expectation}.
\item In the \emph{last} slot of a round, i.e., in slot $4k+1$, the sets and elements do the following:
     Each \emph{set} that received a \textsc{Beep} in at least $3k$ slots adds itself to the solution and beeps. 
     Each uncovered \emph{element} (active or not) listens and considers itself covered if a neighbor joined the solution, i.e., if it hears a \textsc{Beep}.
    If an element is covered, it does not need to send or receive messages for the remainder of the algorithm.
\end{enumerate}
Note that the last step of each round where the sets add themselves is the main difference to the classical algorithm by Jia et al. \cite{DBLP:journals/dc/JiaRS02}. In their work, each candidate locally computed a probability to join the solution. This used information from its 2-neighborhood. 
This approach seems not to be trivially doable with only \textsc{Beep}s.
Finally, note that this algorithm can also be used to compute a \textsc{DominatingSet}, we provide the details at the end of this section.

\begin{algorithm}[ht]
\begin{algorithmic}

    \Procedure{Sets}{$s,N_s,k,\Delta$} \Comment{Executed by the sets $s \in V_{\mathcal{S}}$}
    \For{$i := \Delta, \nicefrac{\Delta}{\sqrt[k]{\Delta}}, \ldots, 0$} \Comment{$k$ phases.}
    \State Pick $X_s \sim Geo(1-\frac{1}{\sqrt[k]{\Delta}})$. \Comment{$Geo(p)$ is the geometric distribution with success prob. $p$.}
    \State Round $X_s$ down to $\Phi := 4k$ if necessary.
    \State Wait $\Phi-X_s$ \textbf{rounds}.
    \For{$\ell=1, \dots, 4k$} \Comment{$4k$ slots.}
        \State Listen for \textsc{Beeps}
    \EndFor
    \If{$s$ receives more than $3k$ \textsc{Beeps}}
    \State Add $s$ to $S$.\Comment{Last slot of the round}
    \State Announce to neighbors via \textsc{Beep}. \Comment{Each round takes $4k+1$ slots.}
    \EndIf
    \State Wait $X_s$ \textbf{rounds} until end of phase.  \Comment{Note that $\Phi = (\Phi-X_s)+X_s$.}
    \EndFor \Comment{Each phase takes $k$ rounds.}
    \EndProcedure
    \State 
    \Procedure{Elements}{$e,N_e,k,\Delta$} \Comment{Executed by the elements $e \in V_{\mathcal{U}}$}
    \For{$i := \Delta, \nicefrac{\Delta}{\sqrt[k]{\Delta}}, \ldots, 0$} \Comment{$k$ phases.}
    \State $Y_e \sim B(\frac{4k}{\Delta_i})$ \Comment{$B(p)$ returns $1$ with prob. $p$}
    \State $a_e \sim Uni(1, \ldots, 4k)$ \Comment{$Uni(a, \ldots, b)$ picks an integer in $[a,b]$ uniformly at random}
    \For{$j :=0, \ldots, 4k$} \Comment{$4k$ rounds.}
        \For{$\ell=1, \dots, 4k$} \Comment{$4k$ slots.}
            \If{($Y_e=1$) \emph{and} ($a_e = \ell$) \emph{and} ($e$ is uncovered)}
            \State Send \textsc{BEEP} to all neighboring sets.
    \EndIf 
    \EndFor
    \State Listen if a neighbor added itself. \Comment{Last slot of the round}
    \State Set $e$ as covered, if it received a \textsc{Beep}  \Comment{Each round takes $4k+1$ slots.}
    \EndFor
    \EndFor \Comment{Each phase takes $4k$ rounds.}
    \EndProcedure
\end{algorithmic}
\caption{\textsc{Beep-and-Sleep (based on \cite{DBLP:journals/dc/JiaRS02}) }}
\label{alg:top_sampling}
\end{algorithm}

\subsection{Analysis}

We will now prove Theorem 1 and show that our algorithm indeed fulfills the promised bounds.
Since the runtime is deterministic, we will only prove the expected approximation ratio. 
\begin{lemma}
\label{lemma:beep_approx}
The algorithm outputs a $O\left(\log^2(\Delta) \cdot \left(\sqrt[k]{\Delta}\right)^3\right)$-approximate solution in expectation.
\end{lemma}
Since the algorithm, for the most part, follows the LRG algorithm by Jia et al.\cite{DBLP:journals/dc/JiaRS02}, a generalized version of their core lemma also holds in this case.
However, we need to "parameterize" it further to account that we cannot precisely count the uncovered elements.
For each element covered by our algorithm, we define the random variables $\eta(u)$ and $\mu(u)$.
We define $\eta(u)$ to be ratio between the  \emph{best} set in $u$'s neighbor and the worst set picked by the algorithm, i.e., the ratio by which the choice of our algorithm differs from the greedy solution. 
Further, let $\mu(u)$ be the random variable that denotes the number of candidates covering $u \in V$ in the round where it is first covered.
Finally, we define $\eta$ and $\mu$ the corresponding variables of the elements that maximize their expected values, i.e.,:
\begin{align}
    E[\mu] = \max_{u \in U} E[\mu(u)] \,\,\,\, \textit{ and } \,\,\,\, E[\eta] = \max_{u \in U} E[\eta(u)]
\end{align}

Given these definitions, the following modified version of Jia et al.'s algorithm\cite{DBLP:journals/dc/JiaRS02} holds:
\begin{lemma}
Let $S$ be the size of the set output by our algorithm. 
Then, it holds. 
\begin{align}
    E[S] \leq  E\left[\eta\right] \cdot E\left[\mu\right] \cdot \log(\Delta) \cdot |S_{OPT}| 
\end{align}
\end{lemma}
\begin{proof}
Before we go into the details of the proof, we define the functions $c_{min}(u)$ and $c_{max}(u)$ for each covered element $u \in V_{\mathcal{U}}$.
Suppose $u$ is covered by sets $s_1, \dots, s_m$, i.e., these sets add themselves simultaneously.
Further, let $d(s_1), \ldots, d(s_m)$ be \emph{spans} of $s_1, \dots, s_m$, i.e., the number of uncovered elements neighboring $s_1, \dots, s_m$.
Based on this, we define $c_{max}$ based on the set with \emph{fewest} uncovered elements,i.e, the set that devaites the most from whatever set the greedy solution would have picked. 
The value $c_{min}$ on the other hand is determined by the set in $u$'s neighborhood with biggest possible span, i.e., the set that the greedy solution would have picked.
Note that this set may not be part of $s_1, \dots, s_m$.
Formally, we have:
\begin{align}
    c_{max}(u) = \max_{v_i \in \{s_1, \dots, s_m\}} \frac{1}{d(s_i)}  &&\textit{and}&& c_{min}(u) = \min_{s \in N_u} \frac{1}{d(s)}   
\end{align}
Further, by the analysis of that greedy algorithm, it holds that:
\begin{align}
    \sum_{u \in V_{\mathcal{U}}} c_{min}(u) \leq H_{\Delta} |S_{OPT}|
\end{align}
On the other hand, by the definition of $\eta(u)$, it holds: 
\begin{align}
    c_{max}(u) = \eta(u)\cdot c_{min}(u)
\end{align}
Given these definitions, we now consider a single round $i$.
In the following, let $S_i$ be set of candidates that add themselves to $S$ in round $i$.
Let $V_i$ denote the set of elements that are uncovered at the start of round $i$.
We have:
\begin{align}
|S_i| &\leq \sum_{v \in S_i} \frac{d(v)}{d(v)} \leq \sum_{v \in S} |C(v)| \frac{1}{d(v)}\\
&\leq \sum_{v \in S} \sum_{u \in C(v)} c_{max}(u)\\
&=  \sum_{u \in V_i} c_{max}(u)\mu(u) \leq \sum_{u \in V_i} c_{min}(u)\cdot\eta(u)\cdot\mu(u)
\end{align}
For the expected value of $S_i$, this implies:
\begin{align}
    E\left[|S_i|\right] &\leq \sum_{u \in V_i} c_{min}(u)E[\eta(u)\cdot\mu(u)]\\
    &\leq  \sum_{u \in V_i} c_{min}(u) \cdot E[\eta(u)] \cdot \Pr{[\mu(u)>0]}E[\mu(u) \mid \mu(u)>0]\\
    &\leq  \sum_{u \in V_i} c_{min}(u) \cdot E[\eta(u)] \cdot E[\mu(u) \mid \mu(u)>0]
\end{align}
The second line follows from the fact that the random events that determine $\eta(u)$ and $\mu(u)$ are independent.
Finally, we can use the linearity of expectation to sum over all $k^2$ rounds and get:
\begin{align}
    E\left[|S|\right] &= \sum_{i=1}^{k^2} E[|S_i|]
    \leq  \sum_{u \in V} E[\eta(u)] \cdot E[\mu(u) \mid t(u)>0] c_{min}(u)
    \leq E[\eta] \cdot E[\mu] \cdot H_{\Delta} |S_{OPT}|
\end{align}
This proves the lemma.
\end{proof}

Given this lemma, it remains to bound $E[{\eta}]$ and $E[\mu]$. 
We begin with the latter.
Our goal is to use a result by Miller et al. \cite{DBLP:conf/spaa/MillerPVX15} (which needs to be adapted for geometric values). They showed that the number of candidates that pick the earliest possible wake-up time is $\sqrt[k]{\Delta}$ in expectation.
Similarly, we can show the following result\footnote{Note that that the analysis in \cite{DBLP:conf/soda/GrunauMRV20} only bounds this term in expectation (and not exact probability) and does not parameterize it on the number of phases as we do. \textbf{We need both these aspects for our problem}. Adding these two aspects to their analysis does not seem to be straightforward as they use a rather complex term to bound the expectation, whereas we can just exploit the fundamental properties of the geometric distribution.}.
{
\begin{lemma}
\label{lemma:prob_mu}
For any value $t>1$ and node $u \in V$, it holds $\mathbf{Pr}\left[\mu_u > t\right] \leq \left(1-\nicefrac{1}{\sqrt[k]{\Delta}}\right)^t$ 
\end{lemma}
}
\begin{proof}

For each node $u \in V$ we define the random variable $\mathcal{P}$ that denotes the \emph{phase} in which $u$ is covered.
Recall that every element is \emph{always} covered because in the last round of the last phase, all remaining sets that cover at least one element simply add themselves to the solution.
Thus, the variable $\mathcal{P}$ takes values in $0, \ldots, 4k$. 
In the remainder, we will focus only on a single phase.
Throughout this proof, we divide the sets into two subsets.
First, let $N_u$ be the set of sets in $u$'s neighborhood, i.e., the sets that can cover $u$. 
We call these the \emph{neighboring candidates}.
Second, let $NN := V_S \setminus N_u$ be all other sets. 
We call these the \emph{non-neighboring candidates}.
Further, we define the wakeup time $\Phi_s$ of each $s \in V_S$ as $\Phi_s := 4k - X_s$.
Now we observe the round in whcih $u$ is covered and see:
\begin{lemma}
Suppose gets $u$ gets covered in the first round of phase $j$, then
\begin{align}
    \pr{\mu_j \geq t \mid u \textit{ is coverd in round } 0} \leq \left(1-\frac{1}{\sqrt[k]{\Delta}}\right)^t
\end{align}
\end{lemma}
\begin{proof}
In this case, the values $X_s$ of \emph{all} candidates that cover $u$ are bigger than or equal to $4k$.
Otherwise, they would have added themselves later.
We call these the \emph{early} candidates.
For a single candidate $s \in N_u$, the probability to be \emph{early} is at most 
\begin{align}
  \pr{s \textit{ is early}} = \sum_{i=0}^\infty \frac{1}{\sqrt[k]{\Delta}}^{4k+i} = \frac{1}{\Delta^4} \sum_{i=0}^\infty \frac{1}{\sqrt[k]{\Delta}}^{i} \leq \frac{2}{\Delta^4}
\end{align}
This follows directly from the definition of the geometric distribution.
As there are \emph{at most} $\Delta$ neighboring candidates that cover $u$, the expected number of early candidates is at most $\frac{2}{\Delta^3}$ as it holds that:
\begin{align}
    E[\textit{Early sets that cover } u] = \sum_{s \in N_u} \pr{s \textit{ is early}} = \sum_{s \in N_u} \frac{2}{\Delta^4} \leq \frac{2\Delta}{\Delta^4} = \frac{2}{\Delta^3}
\end{align}
As $\Delta \geq 2$ this further simplifies to $\frac{1}{\Delta^2}$. 
Since all candidates pick their wakeup time independently and $\mu_u$ is the exactly the number of early candidates, we can use Chernoff to show that:
\begin{align}
    \pr{\mu_j \geq t \mid u \textit{ is coverd in round } 0} &= \pr{\mu_u \geq t\Delta^2\frac{1}{\Delta^2}} \leq \pr{\mu_u \geq t \Delta^2 E[\mu_u]}\\
    &\leq e^{\frac{-t\Delta^2}{3}} \leq e^{-t} \leq \left(1-\frac{1}{\sqrt[k]{\Delta}}\right)^t
\end{align}
This was to be shown.
\end{proof}

The interesting stuff happens, if $u$ gets covered in any other round of phase $i$.
This part is more complex than the previous one as it requires a more careful analysis of the starting times.
First, we will fix all of the random decisions made the algorithm \textbf{except} the decision of the neighboring sets. 
In particular, we condition on the following three variables:
\begin{enumerate}
    \item \textbf{The wakeup times of the non-neighboring sets }$Z_{NN}$. 
    It holds $Z_{NN} := (\phi_1, \ldots, \phi_m')$. 
    For each $l \in [0,m']$, the non-neighboring set $s_l$ wakes up in round $\phi_l$.
    
    \item \textbf{The \textsc{Beeps} by all uncovered elements } $Z_{U}$.
    
    This random variable contains all random choices made by the uncovered elements, i.e, $Z_{U} := \{(Y_e,a_e) \mid e \in V_{\mathcal{U}}\}$.
    Here, the variable $Y_e$ denotes whether $e$ is active and beeps and $a_e$ denotes the slot in which $e$ beeps.
    Based on this, we can define the \emph{number} of slots in which a set $s$ receives a \textsc{Beep} from an element in subset $U \subseteq V_{\mathcal{U}}$ as follows:
    \begin{align}
        \textsc{Slots}(s,U) := \left|\left\{ i \in [0,4k] | \exists e \in U \cap N_s : Y_e = 1 \wedge a_e = i \right\}\right|
    \end{align}
.
    
    \item \textbf{The initial set of uncovered elements }$\mathcal{U}_0$.
    
    This is the set of uncovered elements in round $0$ of phase $j$.
\end{enumerate}
If we fix all these choices, the decision if a neighboring candidate of $u$ adds itself to the solution that depends on \emph{solely} its wake-up time.
More precisely, the random choices of the other nodes define the last possible round in which a candidate may add itself to the solution.
If it wakes up before this round, the concrete round is \emph{only} determined by the geometric distribution. If it wakes up after this round, it will not add itself.

Now we claim the following
\begin{claim}
Given $(Z_{NN},Z_U,\mathcal{U}_0)$ for each neighboring set $s \in N_u$ there is a value $\rho_s \in [0,4k]$ that denotes the last round in which $s$ can cover $u$, i.e., $\pr{s \textit{ covers } u \textit{ first } \mid \Phi_s \geq \rho_s } = 0$ 
\end{claim}
The idea behind the construction is that the \emph{fixed} values $\mathcal{U}_0$, $Z_{NCC}$, and $Z_U$ can clearly define the sets $\mathcal{A}_0 \subset V_{\mathcal{S}}$ that add themselves to the solution in round $0$. 
Given $Z_{NN}$, we see which sets wake up and count the \textsc{Beep}s of their uncovered elements. 
These \textsc{Beeps} are based on $\mathcal{U}_0$ and $Z_U$, so we see if enough elements beep in distinct slots.
In particular, the set must receive a \textsc{Beep} in at least $3k$ slots, so it must hold that
\begin{align}
    \mathcal{A}_0 := \left\{s \not\in N_u \mid \phi_i = 0 \wedge \left(|\textsc{Slots}(s,\mathcal{U}_0)| \geq 3k\right)\right\}
\end{align}
Given $\mathcal{A}_0$, i.e., the sets that add themselves in round $0$, we can then compute $\mathcal{U}_1 \subseteq \mathcal{U}_0$, i.e., all uncovered elements in round $1$.
Then, by repeating the construction, we can use $\mathcal{U}_1$, $Z_{NCC}$, and $Z_U$ to determine $\mathcal{A}_1$ and therefore $\mathcal{U}_2$ and so on.
This can be continued until round $\tau$ by the following recursive formulas:
\begin{align}
    \mathcal{A}_t = \left\{s \not\in N_u \mid \phi_i = 0 \wedge \left(|\textsc{Slots}(s,\mathcal{U}_t)| \geq 3k\right)\right\}
\end{align}
and 
\begin{align}
    U_t = \left\{e \in \mathcal{U}\setminus\{u\} \mid \exists s \in \bigcup_{i=0}^{t-1} \mathcal{A}_i \right\}
\end{align}
This yields the uncovered elements $\mathcal{U}_0, \ldots, \mathcal{U}_{\tau}$. 
For each candidate $s \in N_u$, we can then clearly identify the first round $\rho_s$ where $s$ does not receive enough \textsc{Beeps} to add itself:
\begin{align}
    \rho_s := arg\max_{0 \leq \tau \leq 4k} \textsc{Slots}(s,\mathcal{U}_\tau) \geq 3k
\end{align}


Given this observation, we can map each outcome of $Z_{NCC},Z_U,$ and $\mathcal{U}_0$ to
a collection of thresholds $\overline{\rho} := (\rho_1, \ldots, \rho_{\Delta_u})$ with the properties above.
If we now condition on $\overline{\rho}$ and $\mathcal{P} = j$, we get --- similar to Miller et al. in \cite{DBLP:conf/spaa/MillerPVX15} --- that:
\begin{lemma}
For any possible realization of thresholds $\overline{\rho}$ it holds:
$\pr{\mu_j \leq t \mid \overline{\rho}} \leq \left(1-1/{\sqrt[k]{\Delta}}\right)^t$
\end{lemma}
\begin{proof}
For each neighboring set $s_\ell \in N_u$ define the adapted wakeup time as:
\begin{align}
\Phi'_\ell := \begin{cases}
(\Phi-X_\ell) & \textit{if } (\Phi-X_\ell) \geq \rho_\ell\\
\infty& \textit{else}
\end{cases}    
\end{align}
Here $X_\ell$ is the geometric random variable drawn to determine the wake-up time.

We can now order the these adapted wakeup times $\Phi'_{(1)}, \dots, \Phi'_{(m)}$ such that $\Phi'_{(1)}$ is the earliest wakeup and $\Phi'_{(\Delta_u)}$ is the last.
For each of these ordered wakeup times $\Phi_{(i)}$ we define $c_{(1)}$ to be the candidate that achieves this time, $X_{(i)}$ is the variables drawn by this set, and $\rho_{(i)}$ is the threshold. Given that $\Phi_{(i)} \leq \infty$, it holds:
\begin{align}
    \Phi'_{(i)} := \Phi-X_{(i)}
\end{align}
So the variable $X_{(1)}$ is \emph{not} be the smallest variable that any neighboring candidate drew, but instead the smallest variable by any neighboring candidate \emph{with finite wake-up time}.

Given this definition, we claim the following:
\begin{claim}
The candidate $c_{(1)}$ archiving $\Phi'_{(1)}$ covers $u$ if and only if $\Phi'_{(1)} \leq \infty$.
\end{claim}
\begin{proof}
This claim can easily be verified, by observing the two possibilities.
If $\Phi'_{(1)} \leq \infty$, then $\Phi'_{(1)}$ is earliest round where some neighborhood candidate tries to itself and is still a candidate. This follows directly from the definition of $\Phi'_{(1)}$.
Thus, $c_{(1)}$ must cover $u$ as no candidate could have done it before.
Otherwise, if $\Phi'_{(1)} = \infty$, then it must holds that $\Phi-X_{(1)} \geq \rho_{(1)}$.
Moreover, this implies that $\Phi - X_{(i)} \geq \rho_{(i)}$ for all other candidates as well.
This follows from the fact that $\Phi'_{(1)}$ is the smallest wake-up time by definition, so all others must be infinity, too. 
In this case, \emph{all} potential candidates do not cover enough elements when they wake up. 
This follows directly from the definition of each $\rho_{(i)}$.
Thus, in this case, $u$ is not covered by any candidate in this phase.
\end{proof}
Thus, for $\mu_u = t$ the $t$ smallest values must all be equal and non-infinity.
Formally:
\begin{align}
    \pr{\mu_j = t} = \pr{\Phi'_{(1)}= \dots = \Phi'_{(t)} \mid \Phi_{(1)} \leq \infty}
\end{align}
What follows are some fundamental calculations based on Miller et al.'s proof for exponential random variables.
However, we need to adapt them to geometric variables. 
First, we observe that it holds:
\begin{align}
    \pr{(\Phi - X_\ell) \leq \rho_\ell-x} = \pr{-X_\ell \leq \rho_\ell - \Phi - x} = \pr{X_\ell \geq (\Phi - \rho_\ell) + x}
\end{align}
And therefore, for any $x \leq \rho_\ell$ we see that the wakeup time is only determined by the parameter of the geometric distribution, namely
\begin{align}
    \pr{\Phi'_\ell < x \mid \Phi'_\ell \leq x} &= \pr{(\Phi - X_\ell) \leq x) \mid  \Phi'_\ell \leq x}
    = \pr{X_\ell > \Phi - x \mid \Phi'_\ell \leq x}\\
    &= \pr{X_\ell > \Phi - x \mid X_k \geq \Phi - x}
\end{align}
Now we can use the fact of the geometric distribution is memoryless and see:
\begin{align}
    \pr{\Phi'_\ell < x \mid \Phi'_\ell \leq x} &= \pr{X_\ell > y \mid X_\ell \geq y}& \rhd \textit{Substituting }y:=\Phi-x\\
    &= \pr{X_\ell > 0} = \left(\frac{1}{\sqrt[k]{\Delta}})\right)
\end{align}
Therefore, we can conclude the following for two consecutive $\Phi_{(i)}$ and $\Phi_{(i+1)}$ that it holds:
\begin{align}
    \pr{\Phi'_{(i)} < x  | \Phi'_{(i+1)} = x} &=  \pr{\Phi_{(i)} \leq x+1  | \Phi_{(i+1)} = x} \\
    &= \pr{X_{(i)} \geq \Phi - x+1  | X_{(i)} \geq \Phi - x}\\
    &= \pr{X_{(i)} \geq 1 }= \left(\frac{1}{\sqrt[k]{\Delta}})\right)
\end{align}
And thus, for the opposite event it holds:
\begin{align}
    \pr{\Phi'_{(i)} = x  | \Phi'_{(i+1)} = x} &= \left(1-\frac{1}{\sqrt[k]{\Delta}})\right)
\end{align}
Finally, we condition on $\Phi'_{(t)} = \tau$ for a round $0 \leq \tau \leq 4\log(\Delta)$.
Using the chain rule of conditional probability, we see that:
\begin{align}
    \pr{\Phi'_{(1)}, \ldots, \Phi'_{(t)} = \tau} = \prod_{i=1}^t \pr{X_{(i)}=\tau \mid X_{(i+1)} = \tau } \leq \prod_{i=1}^t \left(1-\frac{1}{\sqrt[k]{\Delta}})\right)= \left(1-\frac{1}{\sqrt[k]{\Delta}})\right)^t 
\end{align}
Note that this is independent of the actual round, the law of total probability yields the the result.
It holds:
\begin{align}
    \pr{\mu_u \geq t} &:= \sum_{\tau = 1}^\Delta \pr{\Phi'_{(t)} = \tau}\pr{\Phi'_{(1)}, \ldots, \Phi'_{(t)} = \tau}\\
    &=  \sum_{\tau = 1}^\Delta \pr{\Phi'_{(t)} = \tau} \left(1-\frac{1}{\sqrt[k]{\Delta}})\right)^t = \left(1-\frac{1}{\sqrt[k]{\Delta}})\right)^t
\end{align}
\end{proof}
As the concrete value of the $\rho_i$'s and the phase is immaterial, by the law of total probability and the lemma follows for a single phase.
\end{proof}
Thus, the geometric series implies $E\left[\mu\right] \leq \sum_{t=1}^{\infty} (1-\nicefrac{1}{2})^t = \frac{1}{1-\nicefrac{1}{\sqrt[k]{\Delta}}} = \sqrt[k]{\Delta}$.
Further, we show that:
\begin{lemma}
For $k \geq 3$ it holds $E[\eta] \leq 25 \cdot \log(\Delta) \cdot \left(\sqrt[k]{\Delta}\right)^2$
\end{lemma}
\begin{proof}
In the following, we fix an element $u \in \mathcal{U}$.
To proof the lemma, we observe the event that there is any set with span bigger than $6 \cdot \log(\Delta) \cdot \sqrt[k]{\Delta} \cdot \Delta_i$ \emph{and} any set smaller than $\frac{\Delta_i}{4 \cdot \sqrt[k]{\Delta}}$ joins the solution. 
Denote this event as $\mathcal{B}$. 
Given that $\pr{\mathcal{B}} \leq \nicefrac{1}{\Delta}$, the law of total expectation implies:
\begin{align}
    E[\eta(u)] &= \pr{\mathcal{B}} \cdot E[\eta(u)\mid\mathcal{B}] + \pr{\mathcal{\neg B}} \cdot E[\eta(u)\mid \neg \mathcal{B}]\\
    &\leq \left(1-\frac{1}{\Delta}\right) 24 \cdot \log(\Delta)  \cdot \left(\sqrt[k]{\Delta}\right)^2 + \frac{2}{\Delta} \cdot \Delta\\
    &\leq 25  \cdot \log(\Delta) \cdot \left(\sqrt[k]{\Delta}\right)^2
\end{align}
Thus, in the following, we show that $\pr{\mathcal{B}}$ is at most $\frac{1}{\Delta}$.
Let $H^i_u \subseteq N_u$ denote all neighbors that span $20 \cdot \log(\Delta) \cdot \sqrt[k]{\Delta} \cdot \Delta_i$ uncovered elements at any point in phase $i$ and likewise let $L^i_u \subseteq N_u$ denote all neighbors that span less than $\frac{\Delta_i}{4 \cdot \sqrt[k]{\Delta}}$ at any point in the phase.
Then, it holds via union bound that:
\begin{align}
    \pr{\mathcal{B}} &\leq \pr{\exists s \in H_u} + \pr{\exists s \in L_u: \textit{$s$ adds itself to the solution}} 
\end{align}
First, we observe that $H_u^i$ is empty with prob. $\nicefrac{1}{\Delta}$.
The proof is straightforward: Suppose there is a set with $c \cdot \log(\Delta) \cdot \sqrt[k]{\Delta} \cdot \Delta_i$ uncovered neighbors in phase $i$. 
Then it had at least $c \cdot \log(\Delta) \cdot \Delta_{i-1}$ uncovered neighbors when it woke up in phase $i-1$.
The probability that in a fixed slot (of the $4k$ slots of that round), no elements beeps is:
\begin{align}
\label{eq:beep_in_slot}
    \pr{\textit{No \textsc{Beep} in slot $j$}} &\leq  \left(1-\frac{1}{4k} \cdot \frac{4k}{\Delta_{i-1}}\right)^{c\cdot \log(\Delta)\cdot\Delta_{i-1}} \leq e^{-c\log(\Delta)} = \frac{1}{\Delta^{c}}
\end{align}
If $c \geq 6$, a union bound over all $4k$ slots then yields that all that with probability higher than $1-\frac{1}{\Delta^2}$, there was at least one beeping element in every slot when $s$ woke up. In particular, there was beeping element in more than $3k$ slots.
Thus, the set in question must have added itself to the solution in phase $i-1$ and it holds:
\begin{align}
    \pr{s \in H^i_u : \textit{$s$ was not added in phase $i-1$}} \leq \sum_{j=1}^{4k} \pr{\textit{No \textsc{Beep} in slot $j$}} \leq \frac{1}{\Delta^{c-4}} \leq \frac{1}{\Delta^2}
\end{align}
Now consider a set $s \in L^i_u$.
The probability that $s$ has $3k$ beeping neighbors on wake-up is bounded as follows: 
\begin{align}
    \pr{\textit{$s$ has $\geq 3k$ beeping neighbors}} &\leq \sum_{i=3k}^{\Delta} \binom{\frac{\Delta_i}{4\sqrt[k]{\Delta}}}{i} \left(\frac{4k}{\Delta_i}\right)^{3k} \leq \Delta \binom{\frac{\Delta_i}{4\sqrt[k]{\Delta}}}{3k} \left(\frac{4k)}{\Delta_i}\right)^{3k}\\
    &\leq \Delta \left(\frac{e\cdot\Delta_i}{4 \sqrt[k]{\Delta} \cdot {3k}}\right)^{3k} \frac{4k^{3k}}{\Delta_i^{3k}} 
    \leq \frac{1}{\Delta^2}
\end{align}
Thus, the probability that $s$ has beeping neighbors in $3k$ \emph{distinct} slots (and therefore adds itself to the solution) can only be smaller.
Therefore, it holds 
\begin{align}
\pr{\mathcal{B}} &\leq \pr{\exists s \in H_u} + \pr{\exists s \in L_u: \textit{$s$ adds itself to the solution}} \\
    &\leq \pr{\exists s \in H_u} + \pr{\exists s \in L_u: \textit{$s$ has $3k$ beeping neighbors}} \\
    &\leq \sum_{s \in H_s} \pr{\textit{$s$ was not added in phase $i-1$}} + \sum_{s \in L_s}\pr{\textit{$s$ has $3k$ beeping neighbors}}\\
    &\leq 2 \Delta\frac{1}{\Delta^2} = \frac{2}{\Delta}
\end{align}
This is what we wanted to show.
\end{proof}
This proves the theorem as $E[\eta\cdot\mu] \in O\left(\log(\Delta)^2\sqrt[k]{\Delta}^3\right)$.

\subsection{Extension to \textsc{DominatingSet}}
Note that the algorithm can also be extended to the \textsc{DominatingSet} problem, where each set is also an element.
At first glance, all elements could simply also execute the code for the sets.
Note, however, that a node cannot simultaneously beep and listen, which causes some problems with a direct simulation.
In particular, an \emph{active} uncovered elements cannot also listen for \textsc{Beeps} in all slots.
To solve this, we simply add all active elements to the solution. 
In phase $i$, each set that is added to the solution covers $\Omega(\log(n)\Delta_{i-1})$ elements, w.h.p.\footnote{If we simply choose $c$ big enough in Equation \ref{eq:beep_in_slot}, we see that all sets with $\log(n)\Delta_i$ elements are gone, w.h.p)}
Suppose there are $U_i$ uncovered elements, then optimal solution must at least be of size $\Theta(\log(n) \frac{U_i}{\Delta_{i-1}})$
As in phase $i$ there are --- by Chernoff --- only $O(\log(n)\frac{U_i}{\Delta_i})$ active elements, w.h.p. we can just add to the solution without (asymtotically) affecting the approximation ratio. Note that this also implies that for $k=\log(\Delta)$ the total number of \textsc{Beeps} is $\Tilde{O}(|S_{OPT}|)$.

\section{A Low-Message $KT_0$ Algorithm}
\label{sec:KT1}
We now move away from the \textsc{Beeping} model and present our low-message $KT_0$-\textsc{Congest} algorithm.
In this section, we make the following less restrictive assumptions about the model:
\begin{enumerate}
    \item Again, we observe a fixed communication graph $G := (V_S \cup V_{\mathcal{U}} ,E)$ with $V_S=n$ and $V_{\mathcal{U}} = m$ with bidirected communication edges $\{s,e\} \in E$ between a set an its elements. Sets and elements can locally distinguish between their edges through port numbers, but no global identifiers uniquely identify a node. To simplify the presentation, all nodes know $\Delta$ and $\log(n+m)$.
    \item Time proceeds in synchronous \emph{rounds}. Each round, an element or set can send a distinct message of size $O(\log(n))$ to any subset of its neighbors. Messages are received in the next round.
    \item All nodes wake up in the same round.
\end{enumerate}
In other literature, this model is sometimes referred to as the \emph{clean network model}\cite{peleg}.
Note that this model does not intend to represent ad-hoc networks faithfully but instead is used to analyze the message efficiency of distributed algorithms (see .e.g., \cite{DBLP:conf/wdag/GmyrP18} for an overview).
Similar to the \textsc{Beeping} model, it starts with limited knowledge of its neighborhood and must \emph{learn} everything it needs to solve the given problem. 
Given this model, we will show the following theorem:
\begin{theorem}
There is an algorithm in the $KT_0$-\textsc{Congest} model that solves \textsc{SetCover} in time $O(\log^2(\Delta))$, expected approximation ratio of $O(\log(\Delta))$, and sends only $\Tilde{O}(\sqrt{\Delta}(n+m))$ messages, w.h.p., given that all nodes know $\Delta$ and an approximation of $\log(n)$.
\end{theorem}
As it turns out, we will only need a few changes to our already established algorithm to prove this theorem.
Every \textsc{Beeping}-Algorithm also works in the $KT_0$-\textsc{Congest}-Model mentioned above, as the model is less restricted.
To obtain a message efficient algorithm, we only need to make a few minor changes to our algorithm: 
     First, to simplify the presentation, we do not parameterize the algorithm with $k$. Instead, we fix $k=\log(\Delta)$ and only consider this case.
     Second, we do not require the notion of slots anymore as a node can simply count how many of its neighbors beeped in a single round as the messages arrive via distinct channels.
    Each \emph{round} now only consists of precisely one slot.
    In particular, an active element does \emph{not} pick a slot number anymore, but directly sends its \textsc{Beep}.
    Further, instead of beeping \emph{all} neighbors with probability $\frac{2^i}{\Delta}$, each element picks $\frac{\log(n)2^i}{\Delta}$ neighbors uniformly at random and sends a \textsc{Beep}.
     Finally, the most significant change to the algorithm is the following: Instead of executing all phases of the algorithm, we only execute it until phase $\frac{\log(\Delta)}{2}$, i.e., until the active degree of all sets is around $O(\sqrt{\Delta})$. We call this the first \emph{stage} of the algorithm.
    Then, all uncovered elements notify their respective set that they are uncovered.
    Denote these elements as $\mathcal{U}'$.
    Finally, the algorithm continues (almost) as usual for the remaining $\frac{\log(\Delta)}{2}$ phases, but each set that joins the solution notifies \emph{only} the elements in $\mathcal{U}'$ that they are covered. We call this the second \emph{stage} of the algorithm.

We will now prove Theorem $2$.
One can easily verify that most lemmas from our previous analysis are still correct.
The first stage does not differ from our previous algorithm at all.
Further, all nodes that do not receive any message in the second stage of the algorithm (but would have in the original algorithm) are already covered. Thus, they would be idle regardless.
Thus, we only need to show that $E[\eta]$ is still small to prove the approximation ratio.

Before we start, we need the following auxiliary lemma that tightly bound the sets that add themselves to the solution in a phase $i$.
\begin{lemma}
\label{lemma:chernoff_eta}
Suppose each elements picks $c \cdot 8 \cdot \log(n)\frac{2^i}{\Delta}$ active edges uniformly and independently at random. 
Further, each set with at least $c\cdot4\log(n)$ active edges adds itself on wake-up.
Then the following two statements hold w.h.p:
\begin{enumerate}
    \item At the end of phase $i$ there is no set with $\frac{\Delta}{2^i}$ uncovered elements.
    \item Any set that adds itself in round $i$ covers at least $\frac{\Delta}{8 \cdot 2^i}$ elements.
\end{enumerate}
\end{lemma}
Both statements follow through an elementary application of the Chernoff Bound.
Note that this lemma directly implies that $E[\eta] \in O(1)$ as it bounds the \emph{worst} and \emph{best} set in each rounds, w.h.p.
Therefore expected approximation ratio is $O(\log(\Delta))$, which is as good as the sequential greedy solution.
\begin{proof}
Both statement follow from Chernoff, we prove them both separatly.
\begin{enumerate}
    \item Let $s$ be a set with $\frac{\Delta}{2^i}$ uncovered element on wake-up, then the expected number on active edges is 
    \begin{align}
        E[A_s^i] = \frac{\Delta}{2^i} \cdot c8\log(n) \cdot \frac{2^i}{\Delta} = c8\log(n)
    \end{align}
    Thus, the probability that only $c4\log(n)$ edges are active is
    \begin{align}
        \Pr[A_s^i \leq c4\log(n)] \leq \Pr[A_s^i \leq (1-\frac{1}{2})E[A_s^i]] \leq e^{\frac{c8\log(n)}{2\cdot 2^2}} = \frac{1}{n^c} 
    \end{align}
    Since with more uncovered elements, the probability can only be smaller, the statement follows.
     \item Let $s$ be a set with $\frac{\Delta}{8 \cdot 2^i}$ uncovered element on wake-up, then the expected number on active edges is 
    \begin{align}
        E[A_s^i] = \frac{\Delta}{8 \cdot2^i} \cdot c8\log(n) \cdot \frac{2^i}{\Delta} = c\log(n)
    \end{align}
    Thus, the probability that at least $c4\log(n)$ edges are active is
    \begin{align}
        \Pr[A_s^i \leq c4\log(n)] \leq \Pr[A_s^i \leq (1+3)E[A_s^i]] \leq e^{\frac{3^2c\log(n)}{3+2}} \leq \frac{1}{n^c} 
    \end{align}
    Since with fewer uncovered elements, the probability can only be smaller, the statement follows.
\end{enumerate}
\end{proof}

Thus, it only remains to analyze the message complexity.
We prove the message bound for each stage of the algorithm.
We begin with the first stage and show that for it holds:
\begin{lemma}
 Until phase $\frac{\log(\Delta)}{2}$, the nodes send at most $\Tilde{O}(n\sqrt{\Delta})$ messages.
\end{lemma}
\begin{proof}
As each element picks (up to) $c8\log(n)2^i/\Delta$ edges for some constant $c > 0$ in phase $i$, the lemma follows immediately.
The corresponding bound for the sets is less trivial as it does not directly follow from the algorithm.
Here, we need to consider that only sets that add themselves to the solution send messages.
In particular, each set that adds itself sends at most $\Delta$ messages and remains silent otherwise.
Thus, we show that at most $\Tilde{O}(\frac{n}{\sqrt{\Delta}})$ sets add themselves w.h.p.
Note that every set that adds itself covers (at least) $\sqrt{\Delta}/8$ uncovered elements, w.h.p, otherwise it would \emph{not} have added itself.
This follows from the second statement in Lemma \ref{lemma:chernoff_eta}.
On the other hand, each uncovered element is covered by \emph{at most} $O(\log(n))$ sets, w.h.p. 
This follows from choosing $t \geq c\log(n)$ for some $c>0$ in Lemma \ref{lemma:prob_mu}.
Let $S_i \subset V_S$ be the solution in phase $i$ and $C_i$ be the covered elements, then it must hold:
\begin{align}
   |S_i| \cdot \sqrt{\Delta}/8 \leq c \log(n) |C_i| &\Leftrightarrow |S_i| \leq \frac{ 8 c \log(n) |C_i|}{\sqrt{\Delta}} \Leftrightarrow |S_i| \leq \frac{ 8 c \log(n) n}{\sqrt{\Delta}} \in \Tilde{O}(\frac{n}{\sqrt{\Delta}})
\end{align}
In other words, if there are more $\Tilde{O}(\frac{n}{\sqrt{\Delta}})$ sets that added themselves, then there must exist an element covered by more than $c\log(n)$ sets.
This is a contradiction, which implies the lemma
\end{proof}
This lemma concludes the analysis of the algorithm's first stage.
Finally, we need to observe the second stage.
In this stage, the algorithm only uses communication edges adjacent to the set of uncovered elements in phase $\log(\Delta)/2$.
Thus, to determine the message complexity, we only need to count these edges.
Formally, we show:
\begin{lemma}
In phase $\frac{\log(\Delta)}{2}$, each set has at most $O(\sqrt{\Delta})$ uncovered elements.
\end{lemma}
\begin{proof}
The lemma follows directly from the first statement of Lemma \ref{lemma:chernoff_eta} as all sets that have more than $O(\sqrt{\Delta})$ uncovered elements must have added themselves in an earlier round w.h.p.
\end{proof}
Thus, in the remaining $O(\log^2(\Delta))$ rounds of the algorithm, all communication will only take place via these $O(m\sqrt{\Delta})$ edges.
Since at most one message passes each edge in every round, this implies that at most $O(m\sqrt{\Delta}\log^2(\Delta))$ messages are sent, which proves Theorem 2.
\subsection*{Lower Bound}
In this section, we prove a lower bound on the number of messages needed to approximate a solution.
The proof works via a reduction to the sequential case.
Here, it is well known that a large portion of the input, i.e., the connections between the nodes, must be revealed to the algorithm.
In particular, it is known that the following holds:
\begin{lemma}[Lower bound from \cite{DBLP:conf/soda/IndykMRVY18}]
Consider a sequential computation model that allows the following two queries
\begin{itemize}
    \item \textbf{EltOf(i,j)} - Returns the $j^{th}$ element of Set $S_i$ or $\bot$ if there is no such element. 
    \item \textbf{SetOf(i,j)} - Returns the $j^{th}$ set which contrains $e_i$ or $\bot$ if there is no such set.
\end{itemize}
Then, every algorithm that yields $O(1)$-approximation for \textsc{SetCover} needs at least $\Tilde{\Omega}(m\sqrt{n})$ queries on certain graphs.
\end{lemma}
In particular, the lower bound graph in $\cite{DBLP:conf/soda/IndykMRVY18}$ has maximal degree $n$ and it holds $m=n$. Therefore, the bound can be rewritten as $\Tilde{\Omega}((m+n)\sqrt{\Delta})$, which is exactly the message complexity of our algorithm. Now we show that if there is a distributed algorithm with less than $\Tilde{\Omega}(m\sqrt{n})$ messages, it can be turned into a sequential algorithm with less than $\Tilde{\Omega}(m\sqrt{n})$ queries. 
This is, of course, a contradiction to the lemma above.
The proof's main ingredient is the observation that every message that is sent from $v$ along its $j^{th}$ channel can be emulated as looking up $SetOf(v,j)$.
Other than that, the proof is quite technical.
The main result is as follows:

\begin{lemma}
Any algorithm that yields an $O(1)$-approximation for \textsc{SetCover} needs at least $\Tilde{\Omega}(m\sqrt{n})$ messages on certain graphs in $KT_0$ model.
\end{lemma}

\begin{proof}
The lemma follows from the fact that any $KT_0$ algorithm that sends $O(x)$ messages can be simulated with $O(x)$ queries in the sequential model.
First, we create $(m+n)$ objects that store the internal variables of each set and element.
These objects are stored in two arrays $A_S$ and $A_E$, s.t., the object for $e_i$ can be accessed through $A_E[i]$. Analogously, each set $s_j$ can be accessed through $A_S[j]$.
Consider a single round of a CONGEST algorithm:
First, we iterate over all objects and perform the local computations that the set or element would execute in the CONGEST algorithm. Now for each message $m_{ij}$ that $e_i$ sends to its $j^{th}$ set, we use the query SetOf$(i,j)$ to obtain its index $j'$ and then add the message to $A_S[j']$. The same is done vice versa for messages from sets to elements.
After all messages have been handled, $A_E$ and $A_S$ contain the nodes' states and all received messages. With this information, each node's action in the distributed algorithm can flawlessly be simulated.
Thus, the sequential algorithm again iterates over $A_E$ and $A_S$ to compute the next round's messages according to the algorithm. 
Therefore, any distributed algorithm sending $O(x)$ messages can be transformed into a sequential algorithm with $O(x)$ queries. 
Therefore, any lower bound on the queries is also a lower bound on the messages.
\end{proof}

\section{Conclusion and Future Work}
In this work, we presented two different message- and energy-efficient distributed algorithms for \textsc{SetCover} for the \textsc{Beeping} model and the $KT_0$ model.
In future work, it would be interesting to see whether the existence of unique identifers known to all nodes can improve the message complexity or if similar bounds hold. 

\bibliography{lipics-v2019-sample-article.bib}

\end{document}